%% file: main.tex
\documentclass[a4paper,reqno]{article}

\input{packages}
\input{technical} %% Defintion, bibliography, typing shortcuts

\title{On the properties of coframes}
\author{Giovanni Canepa}
\affil{Universit\'e de Gen\`eve, Section de Math\'ematiques\\
\href{mailto:giovanni.canepa.math@gmail.com}{giovanni.canepa.math@gmail.com}}

\begin{document}

\maketitle

\begin{abstract}
    We consider injectivity and surjectivity of some maps on the exterior algebra of isomorphic finite-dimensional vector spaces. We prove the properties of the maps in full generality, for any dimension of the vector space and any subspace. We also draw a connection with the  Palatini--Cartan formulation of General Relativity, for which these maps are of crucial importance. 

\end{abstract}
%\tableofcontents

\section{Introduction}

Let $V$ and $Z$ be two $N$-dimensional vector spaces and let $e:V \rightarrow Z$ be an isomorphism between them ($N\in \mathbb{N}$). This isomorphism can be naturally viewed as an element of $$e \in V^* \otimes Z.$$ 
We call such an element a \emph{coframe}\footnote{This name comes from the physical interpretation described in Section \ref{s:PC_gravity}}. The map $e$ can be extended to the exterior algebra of $V^*$ and $Z$ as follows. For $n,k\leq N-s$, define 
\begin{align}\label{e:def_W_intro1}
    W_{s}^{(n,k)}: \textstyle{\bigwedge^n} V^* \otimes \textstyle{\bigwedge^k} Z & \longrightarrow \textstyle{\bigwedge^{n+s}} V^* \otimes \textstyle{\bigwedge^{k+s}} Z \\
    X  & \longmapsto   X \wedge \underbrace{e \wedge \dots \wedge e}_{s-times}. \nonumber
\end{align}

Let now $V^l \subset V$ be a subspace of $V$ of codimension $l<N$. We can restrict $e$ to this subspace and get an injective map $e|_{V^l}: V^l \rightarrow Z$. As before this map can be viewed as an element of $$e|_{V^l}\in (V^l)^* \otimes Z.$$ We can then generalize the maps \eqref{e:def_W_intro1} as (for $n\leq N-s-l$ and $k\leq N-s$)
\begin{align}\label{e:def_W_intro}
    W_{s}^{l,(n,k)}: \textstyle{\bigwedge^n} (V^l)^* \otimes \textstyle{\bigwedge^k} Z & \longrightarrow \textstyle{\bigwedge^{n+s}} (V^l)^* \otimes \textstyle{\bigwedge^{k+s}} Z \\
    X  & \longmapsto   X \wedge \underbrace{e|_{V^l} \wedge \dots \wedge e|_{V^l}}_{s-times}. \nonumber
\end{align}

The importance of these maps is related to the role that they play in the coframe formalism of General Relativity.
A natural question that arises is under which conditions these maps are injective, surjective or isomorphism. The goal of this note is to provide and prove such conditions. In particular we will prove the following theorem.
\begin{theorem}\label{thm:Main_intro}
    The map $W_{s}^{l (n,k)}$ defined in \eqref{e:def_W_intro} is surjective if and only if 
    \begin{align}
        n+k \geq N - s
    \end{align}
    and it is injective if and only if 
    \begin{align}
        n+k \leq N  -l- s . 
    \end{align}
\end{theorem}

\begin{corollary}
    The map $W_{s}^{l (n,k)}$ is an isomorphism if and only if $l=0$ and $n+k=N-s$.
\end{corollary}

The proof of Theorem \ref{thm:Main_intro} is divided in necessary and sufficient condition with the first rather easy and the second more involved. In particular the latter is proved by reducing the problem to the one of inverting a suitable matrix. 

More in detail, in Section \ref{s:conventions}, we introduce some useful notation and a basis for the vector spaces adapted to the isomorphism $e$. In Section \ref{s:technical_lemmas} we then introduce and prove some technical lemmas about the invertibility of a special class of matrices which will be useful in Section \ref{s:proof_theorem} which contains the core proofs of Theorem \ref{thm:Main_intro}.

Finally, in Section \ref{s:corollaries} we prove some corollaries and in Section \ref{s:PC_gravity} we draw the connection of these results with the theory of Palatini--Cartan General Relativity and show some examples of applications of the main theorem.

\subsection*{Aknowledgements}
I am grateful to M. Capoferri, A. S. Cattaneo, F. Fila-Robattino and M. Schiavina for the interesting discussions and suggestions.
The author acknowledges partial support of SNF Grant No P500PT 203085 and No P5R5PT 222221.

\section{Definitions and conventions} \label{s:conventions}

In order to have a less cumbersome notation, let us introduce some notation for the spaces we are considering. Define
\begin{align*}
    \Omega_l^{n,k}:= \textstyle{\bigwedge^n} (V^l)^* \otimes \textstyle{\bigwedge^k} Z \qquad \text{for } n\leq N-l \text{ and } k \leq N.
\end{align*}
Then on these spaces, the maps $W_{s}^{l (n,k)}$ are defined for $s<N$, $n\leq N-s-l$ and $k\leq N-s$ as follows:
\begin{align}
    W_{s}^{l (n,k)}: \Omega_l^{n,k}  & \longrightarrow \Omega_l^{n+s,k+s} \label{e:definition_W}\\
    X  & \longmapsto   X \wedge \underbrace{e \wedge \dots \wedge e}_{s-times}, \nonumber
\end{align}
where we omitted writing the restriction of $e$ to the corresponding subspace $V^{l}$.

We now define a suitable basis on $Z$, constructed out of the isomorphism $e$, in which we will prove the claims. 

The isomorphism $e$, given a basis of $V$, defines a preferred basis on $Z$. Indeed, let $\{v_1, \dots, v_N\}$ be a basis of $V$, then $\{e(v_1), \dots , e(v_N)\}$ is a basis of $Z$. For simplicity of notation we will denote such vectors $\{e_1, \dots, e_N\}$

In the case of a codimension $l$ subspace $V^l$, we can start from a basis $\{v_1, \dots, v_{N-l}\}$ of $V^l$ and then applying $e$ we get $N-l$ linearly independent vectors $\{e_1, \dots, e_{N-l}\}$. We can then complete this basis with vectors $e_{N-l+1} \dots e_N$.\footnote{Note that we do not distinguish this vectors by changing their name, but by convention we distinguish them by the index.}

Hence, fixed a basis on $V$ (and on subspaces), we  call this basis the \emph{standard basis} of $Z$ and, unless otherwise stated, the components will always be taken with respect to this basis.

Hence, in the standard basis, a generic element $X \in \Omega_{l}^{n,k}$ can be written in components as
    \begin{align}\label{e:generic_components}
        X=\sum_{\substack{\mu_1 \dots \mu_n \\c_1 \dots c_k}} X_{\mu_1 \dots \mu_n}^{c_1 \dots c_k} e_1 \wedge \dots \wedge e_k  (v^*)^{\mu_1}\wedge\dots \wedge (v^*)^{\mu_n}
    \end{align}
    where $1 \leq \mu_1 \dots \mu_n\leq N-l$ and $1 \leq c_1 \dots c_k\leq N$ and $(v^*)$ denotes a dual basis to $v$.
    Note that all the indexes appearing in the same row must be different because of antisymmetry.

    Note also, that with respect to this basis we have
    \begin{align}\label{e:e_components}
        e= \sum_{c, \mu}\delta_\mu^c e_c (v^*)^{\mu}  = \sum_{\mu} e_{\mu} (v^*)^{\mu}
    \end{align}
    where $\delta$ is the Kronecker delta. Indeed we get $e(v_i)= \sum_{\mu} e_{\mu} \delta^{\mu}_i = e_i$ as required. 

    Using \eqref{e:generic_components} and \eqref{e:e_components} we can give an expression for $W_{s}^{l (n,k)}(X)$:
    \begin{align}\label{e:Ws_formula}
        W_{s}^{l (n,k)}(X)= \sum_{\substack{\mu_1 \dots \mu_n\\  c_1 \dots c_k\\\mu_{n+1}\dots \mu_{n+s}}}(-1)^{ns} X_{\mu_1 \dots \mu_n}^{c_1 \dots c_k} e_1 \wedge \dots \wedge e_k \wedge e_{\mu_{n+1}}\wedge \dots e_{\mu_{n+s}}(v^*)^{\mu_1} \wedge \dots (v^*)^{\mu_{n+s}}
    \end{align}
    where now in every summand, $\mu_{n+1}, \dots \mu_{n+s}$ are not equal to any $\mu_j$ and any $c_i$ and the indexes $\mu_1 \dots \mu_{n+s}$ take value in $\{1 \dots N-l\}$.

    In order to have a cleaner notation in the proofs in Section \ref{s:proof_theorem} we will loosely use the notation $ W_{s}^{l (n,k)}(X_{\mu_1 \dots \mu_n}^{c_1 \dots c_k})$  for the image under the map $W_{s}^{l (n,k)}$ of an element $X$ with the only nonzero component $X_{\mu_1 \dots \mu_n}^{c_1 \dots c_k}$.

% The goal of this note is to investigate the surjectivity and injectivity of these maps by varying $N, l, s, n$ and $k$.

% The following theorem gives a necessary and sufficient condition for both the properties, in full generality.

% \begin{theorem}\label{thm:Main}
%     The map $W_{s}^{l (n,k)}$ defined in \eqref{e:definition_W} is surjective if and only if 
%     \begin{align}
%         n+k \geq N - s. 
%     \end{align}
%     The map $W_{s}^{l (n,k)}$ defined in \eqref{e:definition_W} is injective if and only if 
%     \begin{align}
%         n+k \geq N  -l- s . 
%     \end{align}
% \end{theorem}

We are now ready to prove the main Theorem.
We divide the proof in some lemmas, starting from some combinatorial and linear-algebraic ones in Section \ref{s:technical_lemmas} and then we prove injectivity in Section \ref{s:injectivity_Lemmas} and surjectivity in Section \ref{s:surjectivity_lemmas}. Since the presence of many indexes can be a little confusing, we will prove the lemmas by first consider some simpler particular cases (usually $l=0$ and $s=1$) and the generalize the result.

\section{Technical Results}\label{s:technical_lemmas}

\begin{lemma}\label{lem:binomial_dimensions}
    Let $m,s,p \in \mathbb{N}$, $m,p,s \geq 1$, and define $q=m+p-1-s$. Then
    \begin{align*}
        \binom{m+q}{m} \geq \binom{m+q}{m-s}.
    \end{align*}
\end{lemma}

\begin{proof}
Let us prove the claim dividing three cases: $p>s$, $p=s$ and $p<s$.
\begin{itemize}
    \item Let $p>s$, i.e. let $p=s+r$ with $r>0$. Note that $m+q= 2m +p-1-s= 2m-1+r$. Since $r>0$ we have $\frac{r-1}{2}\geq 0$ and hence $m \leq \frac{2m-1+r}{2}$. Similarly, using $s\geq 1$ we get $m-s \leq \frac{2m-1+r}{2}$ and $m-s<m$. Hence for the properties of the binomial coefficient (monotone increasing in the second argument until the half of the first argument) we conclude that for $p>s$
        \begin{align*}
        \binom{m+q}{m} \geq \binom{m+q}{m-s}.
        \end{align*}
    \item Let now $p=s$. In this case we get $m+q=2m-1$ and by the properties of the binomial, $\binom{2m-1}{m}$ is bigger than any other binomial $\binom{2m-1}{n}$ for any $n \neq m$, $0\leq n \leq 2m-1$.
    \item Finally let $p<s$, i.e. let $p=s-r$ with $0 < r \leq s-1$. In this case $m+q= 2m +p-1-s= 2m-1-r$. Recall that $\binom{2m-1-r}{m} = \binom{2m-1-r}{m-r-1}$. As before we get the following inequalities:
    \begin{align*}
        m-r-1 < \frac{2m-r-1}{2} \qquad m-s < \frac{2m-r-1}{2} \qquad m-r-1 \geq  m-s.
    \end{align*}
    Hence we conclude with the same argument of the first case that also for $p<s$
        \begin{align*}
        \binom{m+q}{m} \geq \binom{m+q}{m-s}.
        \end{align*}
\end{itemize}
\end{proof}

In order to prove injectivity or surjectivity of the maps \eqref{e:definition_W} we will reduce the problem to a linear algebra one and we will then have to use right or left invertibility of matrices associated to certain systmes. The following results will give us the necessary tools to do so. In the next lemmas we will need the following coefficients $C_i$ ($i=0,\dots q$) defined iteratively as follows:
\begin{subequations}\label{e:coeff_inverses}
    \begin{align} 
        C_0&= \binom{q+s}{s}^{-1} \\
        C_i&= (-1)^i \binom{q+s-i}{s}^{-1} \sum_{j=0}^{i-1} \binom{i}{j} C_j \binom{q+s-i}{s-i+j}.
    \end{align}
    \end{subequations}

\begin{lemma}\label{lem:solution_system_inj}
    Let $m,s,p \in \mathbb{N}$, $m,p,s \geq 1$, and define $q=m+p-1-s$. Let $K$ be a set of $m+q$ indexes. For every subset with $m$ elements $K_m$ define the following equation
    \begin{align}\label{e:system_equations_inj}
        \sum_{K_{m-s}\subset K_m}x^{k_1 \dots k_{m-s}}=0
    \end{align}
    where the sum runs over all possible subsets $K_{m-s}$ of $K_m$ with $m-s$ elements and $K_{m-s}=\{k_1 \dots k_{m-s}\}$. The system with one such equation for each $K_m \subset K$ has as unique solution $x^{k_1 \dots  k_{m-s}}=0$ for all $k_1 \dots k_{m-s} \in K$.
\end{lemma}

\begin{proof}
    First of all we notice that the system has $\binom{m+q}{m}$ equations (the number of possible $m$-elements subsets of $K$) and $\binom{m+q}{m-s}$ unknowns  (the set of possible $(m-s)$-elements subsets of $K$).  By Lemma \ref{lem:binomial_dimensions} we then get that the system has more equations that variables (or an equal number). Let $A$ be the $\binom{m+q}{m} \times \binom{m+q}{m-s}$-matrix representing the system. A sufficient condition for it to be maximal rank (and hence have as the unique solution the 0 vector) is to find a matrix $\binom{m+q}{m-s} \times \binom{m+q}{m}$-matrix $B$ such that $BA=1_{\binom{m+q}{m-s}}$ where $1_{\binom{m+q}{m-s}}$ is the $\binom{m+q}{m-s}$-dimensional identity.

    Let us denote by 
    \begin{align*}
        y^{k_1 \dots k_{m-s}l_1 \dots l_s}=\sum_{K_{m-s}\subset K_m}x^{k_1 \dots k_{m-s}}=0
    \end{align*}
    where $l_1\dots l_s \in K\setminus K_{m-s}$. Then a right inverse is given by the following system of equations
    \begin{align}\label{e:inj_leftinverse}
        x^{k_1 \dots k_{m-s}}= \sum_{i=0}^{q} \sum_{\substack{ 
        k_j \dots k_{j+i} \in \{k_1 \dots k_{m-s}\}\\
        k_m \dots k_{m+i} \in K \setminus\{k_1 \dots k_{m-s}\} }} C_i y^{k_1 \dots \widehat{k_j}\dots \widehat k_{j+i}\dots k_{m-s} l_1 \dots l_{s+i}}
    \end{align}
    where $\widehat{k_j}$ denotes that the index $k_j$ has been omitted and the coefficients $C_i$ are defined in \eqref{e:coeff_inverses}.
    It is a long but straightforward computation to show that this is actually an inverse to \eqref{e:system_equations_inj}.
    Hence we can conclude that the unique solution to \eqref{e:system_equations_inj} is $x^{k_1 \dots  k_{m-s}}=0$ for all $k_1 \dots k_{m-s} \in K$.
\end{proof}

\begin{lemma}\label{lem:solution_system_surj}
    Let $m,s,p \in \mathbb{N}$, $m,p,s \geq 1$, and define $q=m+p-1-s$. Let $K$ be a set of $m+q$ indexes. For every subset with $q$ elements $K_q$ define the following equation
    \begin{align}\label{e:system_equations_surj}
        x^{K_q}=\sum_{K_{s}\subset K \setminus K_q}y^{K_q K_s}
    \end{align}
    where the sum runs over all possible subsets $K_{s}$ of $K \setminus K_q$ with $s$ elements. The system with one such equation for each $K_q \subset K$ can be right-inverted and the inverse is given by
    \begin{align}\label{e:surj_righinverse}
        y^{k_1 \dots k_{q+s}}= \sum_{i=0}^{q+s} \sum_{\substack{ 
        k_j \dots k_{j+i} \in \{k_1 \dots k_{q+s}\}\\
        h_{1} \dots h_{i} \in K \setminus\{k_1 \dots k_{q+s}\} }} C_i x^{k_1 \dots \widehat{k_j}\dots \widehat k_{j+s+i}\dots k_{q+s} l_1 \dots l_{i}}.
    \end{align}
\end{lemma}

\begin{proof}
    The matrix associated to the system \eqref{e:system_equations_surj} is the transposed of the matrix associated to the system \eqref{e:system_equations_inj}. Hence, using the proof of Lemma \ref{lem:solution_system_inj} we immediately find that the system is right invertible with inverse given by the transpose of \ref{e:inj_leftinverse} which is exactly \eqref{e:surj_righinverse}.
\end{proof}

\section{Proof of Theorem \ref{thm:Main_intro}} \label{s:proof_theorem}
\subsection{Injectivity}\label{s:injectivity_Lemmas}

We divide the proof into necessary and sufficient condition, starting from the first.

\begin{lemma}\label{lem:inj_nec_cond}
    The map $W_{s}^{l (n,k)}$ cannot be injective if 
    \begin{align}
        n+k +l > N -s. 
    \end{align}
\end{lemma}

\begin{proof}
     Let us begin from the case $l=0$ and $s=1$ and consider the map
    \begin{align*}
        W_{1}^{0 (n,k)}: \Omega_{0}^{n,k} \rightarrow \Omega_{0}^{n+1,k+1}.
    \end{align*}
    Suppose $n+k > N-1$. Let us now construct  an element $X\neq 0$,  $X \in \Omega_{0}^{n,k} $ such that $W_{1}^{0 (n,k)}(X)=0$. Since $n+k > N-1$, it is possible to find two set of indexes $I_k=\{c_1, \dots c_k\}$ (with $c_i\neq c_j$ for all $1\leq i,j \leq k$) and $J_n=\{\mu_1 \dots \mu_n\}$ (with $\mu_i\neq \mu_j$ for all $1\leq i,j \leq n$) such that $\{1,\dots N\}= I_k \cup J_n$. Then define $X$ such that its component $X_{J_n}^{I_k}=1$ and all the others are zero. From the formula \eqref{e:Ws_formula} for $s=1$ we see that, since $\mu_{n+1}$ is not equal to any $\mu_j$ and any $c_i$, the component $X_{J_n}^{I_k}$ does not appear and hence $W_{1}^{0 (n,k)}(X)=0.$ This implies that the map $W_{1}^{0 (n,k)}$ is not injective.

    Let us now generalize this construction to $l\geq 0$ and $s\geq 1$ and consider the map
     \begin{align*}
        W_{s}^{l (n,k)}: \Omega_{l}^{n,k} \rightarrow \Omega_{l}^{n+s,k+s}.
    \end{align*}
    As before we will construct an element $X\neq 0$, $X \in \Omega_{l}^{n,k} $ such that $W_{s}^{l (n,k)}(X)=0$. 
    Since $n+k > N-s-l$, it is possible to find two set of indexes $I_k=\{c_1, \dots c_k\}$ and $J_n=\{\mu_1 \dots \mu_n\}$ such that $\{1,\dots N+1-s-l\}= I_k \cup J_n$. Define $X$ such that its component $X_{J_n}^{I_k}=1$ and all the others are zero. We can find at most $s-1$ indexes from the set $\{1,\dots N-l\}$ which are different from every index in $I_k$ and $J_n$.\footnote{Note that by construction of the standard basis in codimension $l$ the components of $X$ and $e$ cannot have the indexes $N-l+1 \dots N$ as a coordinate index (i.e. in the lower row).}

    Since in formula \eqref{e:Ws_formula} the only summands appearing have $s$ indexes $\mu_{n+1}\dots \mu_{n+s}$ not equal to any $\mu_j$ and any $c_i$, the component $X_{J_n}^{I_k}$ does not appear and hence $W_{s}^{l (n,k)}(X)=0.$ This implies that the map $W_{s}^{l (n,k)}$ is not injective.
\end{proof}

As a warm up for the proof of the sufficient condition, let us prove the following lemma.
\begin{lemma}
    Let $n+k=N-1$, $l=0$, $s=1$. Then the maps $W_{1}^{0 (n,k)}$ are injective.
\end{lemma}
\begin{proof}
    Let $X \in \Omega_{0}^{n,k} $ such that $Y= W_{1}^{0 (n,k)}(X)=0$. We have to prove that $X=0$. Fixed a coordinate system on $M$, using the standard basis, the equation $Y=0$ will become a system of $\binom{N}{n+1}\binom{N}{k+1}$ equations, one for each component of $Y$.
    
    By Lemma \ref{lem:inj_nec_cond} every component appear in the system.
    %First we prove that all the components of $X$ appear in this system.
    In particular, since $n+k<N$, for every set of indexes $I_k=\{c_1, \dots c_k\}$ and $J_n=\{\mu_1 \dots \mu_n\}$  there exists a index $1\leq b\leq N$ such that $b \notin \{\mu_1, \dots \mu_n\} \cup \{c_1, \dots c_k\}$, hence the component $X_{J_n}^{I_k}$ appears at least in the equation $$Y_{J_n b}^{I_k b}=0.$$

    If $X_{J_n}^{I_k}$ is the only component appearing in $Y_{J_n b}^{I_k b}=0$, then we immediately get $X_{J_n}^{I_k}=0$. 

    Otherwise let us suppose that exactly $m$ components appear in the equation $Y_{J_{n+1}}^{I_{k+1}}=0$. By construction, this means that $J_{n+1} \cap I_{k+1}$ has exactly $m$ elements $b_1 \dots b_m$. We then have that $J_{n+1} \cup I_{k+1}$ has $n+k+2-m= N+1-m$ elements. Hence there exist $m-1$ indexes, $f_1 \dots f_{m-1} \notin J_{n+1} \cup I_{k+1}$. 
        
    Denoting $J_{n+1}= J_{n+1-m} \cup \{b_1 \dots b_m\}$ and $I_{k+1}= I_{k+1-m} \cup \{b_1 \dots b_m\}$, we then consider the equations corresponding to the $\binom{2m-1}{m}$ components $Y_{J_{n+1-m}K_m}^{I_{k+1-m}K_m}$, one for each possible $m$-element subset $K_m$ of $\{f_1 \dots f_{m-1}, b_1 \dots b_m\}$. This system  has $\binom{2m-1}{m}$ equations for $\binom{2m-1}{m-1}$ variables $X_{J_{n+1-m}K_{m-1}}^{I_{k+1-m}K_{m-1}}$, one for each possible $(m-1)$-element subset $K_{m-1}$ of $\{f_1 \dots f_{m-1}, b_1 \dots b_m\}$. If we denote by $k_i$ the elements of a chosen $K_m$, every equation will have the form
    \begin{align*}
        \sum_{i=1}^{m} X_{J_{n+1-m}k_1 \dots \widehat{k_i}\dots k_m}^{I_{k+1-m}k_1 \dots \widehat{k_i}\dots k_m}=0
    \end{align*}
    where $\widehat{k_i}$ denotes that the index has been removed.
    By Lemma \ref{lem:solution_system_inj}, this system of equations has solution
    $X_{J_{n+1-m}k_1 \dots \widehat{k_i}\dots k_m}^{I_{k+1-m}k_1 \dots \widehat{k_i}\dots k_m}=0$ for every choice $K_{m-1}$ and every $i$.
    In particular we get $X_{J_n}^{I_k}=0$.
\end{proof}

\begin{remark}
    As an example we can consider the case when exactly two components $X_{J_n}^{I_k}$ and $X_{J'_n}^{I'_k}$ appear in the equation $Y_{J_n b}^{I_k b}=0$. By construction, this means that $J_n \cap I_k$ has one element $b'$. Without loss of generality we can assume $c_k=b'= \mu_n$. Call $J_{n-1}=\{\mu_1 \dots \mu_{n-1}\}$ and $I_{k-1}=\{c_1, \dots c_{k-1}\}$. $J_{n-1} \cup I_{k-1}$ has $n+k-2=N-3$ elements, since $n+k=N-1$. Hence, there exists an index $1\leq f \leq N$ such that $f \neq b'$, $f \neq b$, $f \notin J_{n-1} \cup I_{k-1}$.

        Then we can then consider the equations corresponding to the components $Y_{J_{n-1} b b'}^{I_{k-1} b b'}$, $Y_{J_{n-1} b f}^{I_{k-1} b f}$, $Y_{J_{n-1} b' f}^{I_{k-1} b' f}$. The corresponding equations  read:
        \begin{align*}
            Y_{J_{n-1} b b'}^{I_{k-1} b b'}=X_{J_{n-1} b}^{I_{k-1} b} + X_{J_{n-1} b'}^{I_{k-1} b'}=0 \\
            Y_{J_{n-1} b f}^{I_{k-1} b f}=X_{J_{n-1} b}^{I_{k-1} b} + X_{J_{n-1} f}^{I_{k-1} f}=0 \\
            Y_{J_{n-1} b' f}^{I_{k-1} b' f}=X_{J_{n-1} b'}^{I_{k-1} b'} + X_{J_{n-1} f}^{I_{k-1} f}=0. 
        \end{align*}
        In this simple case, it is easy to see that the only possible solution to this system is $X_{J_{n-1} b}^{I_{k-1} b}=X_{J_{n-1} b'}^{I_{k-1} b'}=X_{J_{n-1} f}^{I_{k-1} f}=0$. Otherwise, we can use the general formula \eqref{e:inj_leftinverse} with $m=2$, $p=s=1$ $q=1$.
        Using this numbers we get $C_0=\frac{1}{2}$ and $C_1= -\frac{1}{2}$ and finally
        \begin{align*}
            X_{J_{n-1} b}^{I_{k-1} b} = \frac{1}{2} \left(Y_{J_{n-1} b b'}^{I_{k-1} b b'} + Y_{J_{n-1} b f}^{I_{k-1} b f}- Y_{J_{n-1} b' f}^{I_{k-1} b' f}\right)
        \end{align*}
        and similarly for the other components.
\end{remark}

The proof of the general theorem is done by generalizing this construction for any $n+k< N$, any $l\geq 0$ and any $s\geq 1$.

We first generalize this construction for any $n+k< N$ with the following lemma.
\begin{lemma}
    Let $n+k=N-p$, $p\geq 1$ $l=0$, $s=1$. Then the maps $W_{1}^{0 (n,k)}$ are injective.
\end{lemma}
\begin{proof}
    We proceed as before. Let $X \in \Omega_{0}^{n,k} $ such that $Y= W_{1}^{0 (n,k)}(X)=0$, which using the standard basis is a system of equations. 
    
    Let us now suppose that exactly $m$ components appear in the equation $Y_{J_{n+1}}^{I_{k+1}}=0$. By construction, this means that $J_{n+1} \cap I_{k+1}$ has exactly $m$ elements $b_1 \dots b_m$. We then have that $J_{n+1} \cup I_{k+1}$ has $n+k+2-m= N-p+2-m$ elements. Hence there exist $q=m+p-2$ indexes, $f_1 \dots f_{q} \notin J_{n+1} \cup I_{k+1}$. 
        
        Denoting $J_{n+1}= J_{n+1-m} \cup \{b_1 \dots b_m\}$ and $I_{k+1}= I_{k+1-m} \cup \{b_1 \dots b_m\}$, we then consider the equations corresponding to the $\binom{m+q}{m}$ components $Y_{J_{n+1-m}K_m}^{I_{k+1-m}K_m}$, one for each possible $m$-element subset $K_m$ of $\{f_1 \dots f_{q}, b_1 \dots b_m\}$. This system  has $\binom{m+q}{m}$ equations for $\binom{m+q}{m-1}$ variables $X_{J_{n+1-m}K_{m-1}}^{I_{k+1-m}K_{m-1}}$, one for each possible $(m-1)$-element subset $K_{m-1}$ of $\{f_1 \dots f_{q}, b_1 \dots b_m\}$. If we denote by $k_i$ the elements of a chosen $K_m$, every equation will have the form
        \begin{align*}
            \sum_{i=1}^{m} X_{J_{n+1-m}k_1 \dots \widehat{k_i}\dots k_m}^{I_{k+1-m}k_1 \dots \widehat{k_i}\dots k_m}=0
        \end{align*}
        where $\widehat{k_i}$ denotes that the index has been removed.
        By Lemma \ref{lem:solution_system_inj}, this system of equations has solution
        $X_{J_{n+1-m}k_1 \dots \widehat{k_i}\dots k_m}^{I_{k+1-m}k_1 \dots \widehat{k_i}\dots k_m}=0$ for every choice $K_{m-1}$ and every $i$.
        In particular we get $X_{J_n}^{I_k}=0$.

\end{proof}

The same proof holds for $l>0$ with $n+k=N-l-p$ and choosing $q=m+p-2$. Finally let us $s\geq 1$. We have the following lemma.
\begin{lemma}
    Let $n+k=N-l-s-p+1$, $p\geq 1$. Then the maps $W_{s}^{l (n,k)}$ are injective.
\end{lemma}
\begin{proof}
    We proceed as before. Let $X \in \Omega_{l}^{n,k} $ such that $Y= W_{s}^{l (n,k)}(X)=0$. Fixed a coordinate system on $M$, using the standard basis, the equation $Y=0$ will become a system of $\binom{N-l}{n+s}\binom{N}{k+s}$ equations, one for each component of $Y$.
    
    Since $n+k<N-l-s$, for every set of indexes $I_k=\{c_1, \dots c_k\}$ and $J_n=\{\mu_1 \dots \mu_n\}$  there exist $s$ indexes $1\leq b_1 \dots b_s \leq N-l$ such that $b_1 \dots b_s \notin \{\mu_1, \dots \mu_n\} \cup \{c_1, \dots c_k\}$, hence the component $X_{J_n}^{I_k}$ appears at least in the equation $$Y_{J_n b_1 \dots b_s}^{I_k b_1 \dots b_s}=0.$$

    Let us now suppose that exactly $\binom{m}{s}$ components\footnote{This is the most general case. By construction it is not possible to have equations with a different number of terms.} appear in the equation $Y_{J_{n+s}}^{I_{k+s}}=0$. By construction, this means that $J_{n+s} \cap I_{k+s}$ has exactly $m$ elements $b_1 \dots b_{m}$. We then have that $J_{n+s} \cup I_{k+s}$ has $n+k+2s-m= N-l-s-p+1+2s-m+=N-l-p-m+1+s$ elements. Hence there exist $q=m+p-1-s$ indexes, $f_1 \dots f_{q} \notin J_{n+s} \cup I_{k+s}$. 
        
        Denoting $J_{n+s}= J_{n+s-m} \cup \{b_1 \dots b_m\}$ and $I_{k+s}= I_{k+s-m} \cup \{b_1 \dots b_m\}$, we then consider the equations corresponding to the $\binom{m+q}{m}$ components $Y_{J_{n+s-m}K_m}^{I_{k+s-m}K_m}$, one for each possible $m$-element subset $K_m$ of $\{f_1 \dots f_{q}, b_1 \dots b_m\}$. This system  has $\binom{m+q}{m}$ equations for $\binom{m+q}{m-s}$ variables $X_{J_{n+s-m}K_{m-s}}^{I_{k+s-m}K_{m-s}}$, one for each possible $(m-s)$-element subset $K_{m-s}$ of $\{f_1 \dots f_{q}, b_1 \dots b_m\}$.\footnote{Note that by Lemma \ref{lem:binomial_dimensions} $\binom{m+q}{m}\geq \binom{m+q}{m-s}$ for all $m$ and $q=m+p-1-s$, $p\geq 1$.} If we denote by $k_i$ the elements of a chosen $K_m$, every equation will have the form
        \begin{align*}
            \sum_{i=1}^{m} X_{J_{n+1-m}k_1 \dots \widehat{k_i}\dots k_m}^{I_{k+1-m}k_1 \dots \widehat{k_i}\dots k_m}=0
        \end{align*}
        where $\widehat{k_i}$ denotes that the index has been removed.
        By Lemma \ref{lem:solution_system_inj}, this system of equations has solution
        $X_{J_{n+1-m}k_1 \dots \widehat{k_i}\dots k_m}^{I_{k+1-m}k_1 \dots \widehat{k_i}\dots k_m}=0$ for every choice $K_{m-1}$ and every $i$.
        In particular we get $X_{J_n}^{I_k}=0$.

\end{proof}

\subsection{Surjectivity}\label{s:surjectivity_lemmas}
For surjectivity we proceed analogously as for injectivity, however we will omit some of the easier cases. Let us start from the necessary condition.
\begin{lemma}
    The map $W_{s}^{l (n,k)}$ cannot be surjective if 
    \begin{align}
        n+k < N - s. 
    \end{align}
\end{lemma}

\begin{proof}
Consider the map
     \begin{align*}
        W_{s}^{l (n,k)}: \Omega_{l}^{n,k} \rightarrow \Omega_{l}^{n+s,k+s}.
    \end{align*}
    We want to show that there exists an element $Y \in \Omega_{l}^{n+s,k+s}$ which is not in the image of $W_{s}^{l (n,k)}$. It is sufficient to show that there exists a component which cannot be generated from any $X \in \Omega_{l}^{n,k} $.

    Since $n+k < N-s$, it is possible to find two set of indexes $I_{k+s}=\{c_1, \dots c_{k+s}\}$ and $J_{n+s}=\{\mu_1 \dots \mu_{n+s}\}$ such that $ I_k \cap J_n$ has strictly less than $s$ elements. Define $Y$ such that its component $Y_{J_{n+s}}^{I_{k+s}}=1$ and all the others are zero.
    
    Since in formula \eqref{e:Ws_formula} all the components in the image have $s$ index shared in the two rows\footnote{If we consider the components of the image the indexes on the upper row are those corresponding to the indexes of the $e$'s, while the ones in the lower row are those corresponding to the $dx$'s.} Hence $Y$ is not an image of any $X$. This implies that the map $W_{s}^{l (n,k)}$ is not surjective.
\end{proof}

\begin{lemma}
    Let $n+k=N-1$, $l=0$, $s=1$. Then the maps $W_{1}^{0 (n,k)}$ are injective.
\end{lemma}
\begin{proof}
    Let $Y \in \Omega_{0}^{n+1,k+1}$. We want to prove that there exists an $X\in \Omega_{0}^{n,k}$ such that $Y= W_{1}^{0 (n,k)}(X)$. Fixed a coordinate system on $M$, using the standard basis, the equation $Y=W_{1}^{0 (n,k)}(X)$ will become a system of $\binom{N}{n+1}\binom{N}{k+1}$ equations, one for each component of $Y$.
    
    First we prove that all the components of $Y$ appear in this system. Since $n+1+k+1>N$, for every set of indexes $I_{k+1}=\{c_1, \dots c_{k+1}\}$ and $J_{n+1}=\{\mu_1 \dots \mu_{n+1}\}$  there exists a index $1\leq b\leq N$ such that $b \in \{\mu_1, \dots \mu_{n+1}\} \cap \{c_1, \dots c_{k+1}\}$, hence the component $Y_{J_{n+1}}^{I_{k+1}}$ appears at least as image of $X_{J_n}^{I_k}$ where $J_n = J_{n+1}\setminus \{b\}$ and $I_k = I_{k+1}\setminus \{b\}$.

    Our goal is to prove that there exists a combination of components of $X$ that \emph{generates} this component of $Y$.\footnote{Here by generate we mean that if we construct an element $X\in \Omega_{0}^{n,k}$ with just these non-zero components, then under the map $W_{1}^{0 (n,k)}(X)$ we get the desired component.} We have now various possibilities.

    \begin{itemize}
        \item If the component $X_{J_n}^{I_k}$ is only generating the aforementioned component $Y_{J_{n+1}}^{I_{k+1}}$, then we can choose $X$ to have just this component and conclude for surjectivity. 
        
        % \item Otherwise let us suppose that exactly two components $X_{J_n}^{I_k}$ and $X_{J'_n}^{I'_k}$ appear in the equation $Y_{J_n b}^{I_k b}=0$. By construction, this means that $J_n \cap I_k$ has one element $b'$. Without loss of generality we can assume $c_k=b'= \mu_n$. Call $J_{n-1}=\{\mu_1 \dots \mu_{n-1}\}$ and $I_{k-1}=\{c_1, \dots c_{k-1}\}$. $J_{n-1} \cup I_{k-1}$ has $n+k-2=N-3$ elements, since $n+k=N-1$. Hence, there exists an index $1\leq f_1 \leq N$ such that $f_1 \neq b'$, $f_1 \neq b$, $f_1 \notin J_{n-1} \cup I_{k-1}$.

        % For no reason whatsoever we can then consider the equations corresponding to the components $Y_{J_{n-1} b b'}^{I_{k-1} b b'}$, $Y_{J_{n-1} b f_1}^{I_{k-1} b f_1}$, $Y_{J_{n-1} b' f_1}^{I_{k-1} b' f_1}$. The corresponding equations (from the general formula \eqref{e:W1_formula}) read:
        % \begin{align*}
        %     X_{J_{n-1} b}^{I_{k-1} b} + X_{J_{n-1} b'}^{I_{k-1} b'}=0 \\
        %     X_{J_{n-1} b}^{I_{k-1} b} + X_{J_{n-1} f_1}^{I_{k-1} f_1}=0 \\
        %     X_{J_{n-1} b'}^{I_{k-1} b'} + X_{J_{n-1} f_1}^{I_{k-1} f_1}=0. 
        % \end{align*}
        % In this simple case, it is easy to see that the only possible solution to this system is $X_{J_{n-1} b}^{I_{k-1} b}=X_{J_{n-1} b'}^{I_{k-1} b'}=X_{J_{n-1} f_1}^{I_{k-1} f_1}=0$.

        \item  Suppose now that the component $X_{J_n}^{I_k}$ is generating exactly $m$ components of $Y$ one of which is $Y_{J_{n+1}}^{I_{k+1}}$.
        By construction, this means that there are exactly $m$ indexes $b_1 \dots b_m \notin J_{n} \cup I_{k}$. Since $n+k=N-1$, this implies that $J_{n} \cap I_{k}$ has $q=m-1$ elements $f_1\dots f_{q}$.
        %$n+k-q+m=N=N-1-q+m$

        Denoting $J_{n}= J_{n-q} \cup \{f_1 \dots f_q\}$ and $I_{k}= I_{k-q}\cup \{f_1 \dots f_q\}$, we then consider the  $\binom{m+q}{q}$ components $X_{J_{n-q}k_1 \dots k_{q}}^{I_{k-q}k_1 \dots k_{q}}$ one for each possible $q$-element subset $K_{q}$ of $K=\{f_1 \dots f_{q}, b_1 \dots b_m\}$. Each of these components we consider, with a slight abuse of notation the following equations:
        \begin{align*}
            W_{1}^{0 (n,k)}(X_{J_{n-q}k_1 \dots k_{q}}^{I_{k-q}k_1 \dots k_{q}})= \sum_{k_{q+1} \in K\setminus K_{q}} Y_{J_{n-q}k_1 \dots k_{q}k_{q+1}}^{I_{k-q}k_1 \dots k_{q}k_{q+1}}.
        \end{align*}
        The components in the image are $\binom{m+q}{q+1}$.

        By Lemma \ref{lem:solution_system_surj} it is possible to invert this system and get 
        \begin{align*}
            Y_{J_{n-q} k_1 \dots k_{q+1}}^{I_{k-q} k_1 \dots k_{q+1}}= \sum_{i=0}^{q+1} \sum_{\substack{ k_j \dots k_{j+i} \in \{k_1 \dots k_{q+1}\}\\ h_{1} \dots h_{i} \in K \setminus\{k_1 \dots k_{q+1}\} }} C_i X_{J_{n-q}k_1 \dots \widehat{k_j}\dots \widehat k_{j+1+i}\dots k_{q+1} l_1 \dots l_{i}}^{I_{k-q}k_1 \dots \widehat{k_j}\dots \widehat k_{j+1+i}\dots k_{q+1} l_1 \dots l_{i}}.
        \end{align*}
    \end{itemize}
    Hence, for each component $Y_{J_{n-q} k_1 \dots k_{q+1}}^{I_{k-q} k_1 \dots k_{q+1}}$ we can get a linear combination of components of $X$ that generates exactly it, proving that the map is surjective.
\end{proof}

Let us now generalize this proof to the most general case.

\begin{lemma}
    Let $n+k=N-s+p-1$, $p\geq 1$. Then the maps $W_{s}^{l (n,k)}$ are surjective.
\end{lemma}
\begin{proof}
    Let $Y \in \Omega_{l}^{n+s,k+s}$. We want to prove that there exists an $X\in \Omega_{l}^{n,k}$ such that $Y= W_{s}^{l (n,k)}(X)$. Fixed a coordinate system on $M$, using the standard basis, the equation $Y=W_{s}^{l (n,k)}(X)$ will become a system of $\binom{N}{n+s}\binom{N}{k+s}$ equations, one for each component of $Y$.
    
    First we prove that all the components of $Y$ appear in this system. Since $n+s+k+s=N-s+p-1+2s=N+s+p-1>N+s-1$, for every set of indexes $I_{k+s}=\{c_1, \dots c_{k+s}\}$ and $J_{n+s}=\{\mu_1 \dots \mu_{n+s}\}$  there exist $s$ indexes $1\leq b_1 \dots b_s\leq N$ such that $b_1 \dots b_s\in \{\mu_1, \dots \mu_{n+s}\} \cap \{c_1, \dots c_{k+s}\}$, hence the component $Y_{J_{n+s}}^{I_{k+s}}$ appears at least as image of $X_{J_n}^{I_k}$ where $J_n = J_{n+s}\setminus \{b_1 \dots b_s\}$ and $I_k = I_{k+s}\setminus \{b_1 \dots b_s\}$.

    Our goal is to prove that there exists a combination of components of $X$ that \emph{generates} this component of $Y$.\footnote{Here by generate we mean that if we construct an element $X\in \Omega_{l}^{n,k}$ with just these non-zero components, then under the map $W_{s}^{l (n,k)}(X)$ we get the desired component.} We have now various possibilities.

    \begin{itemize}
        \item If the component $X_{J_n}^{I_k}$ is only generating the aforementioned component $Y_{J_{n+s}}^{I_{k+s}}$, then we can choose $X$ to have just this component and conclude for surjectivity. 
        
        % \item Otherwise let us suppose that exactly two components $X_{J_n}^{I_k}$ and $X_{J'_n}^{I'_k}$ appear in the equation $Y_{J_n b}^{I_k b}=0$. By construction, this means that $J_n \cap I_k$ has one element $b'$. Without loss of generality we can assume $c_k=b'= \mu_n$. Call $J_{n-1}=\{\mu_1 \dots \mu_{n-1}\}$ and $I_{k-1}=\{c_1, \dots c_{k-1}\}$. $J_{n-1} \cup I_{k-1}$ has $n+k-2=N-3$ elements, since $n+k=N-1$. Hence, there exists an index $1\leq f_1 \leq N$ such that $f_1 \neq b'$, $f_1 \neq b$, $f_1 \notin J_{n-1} \cup I_{k-1}$.

        % For no reason whatsoever we can then consider the equations corresponding to the components $Y_{J_{n-1} b b'}^{I_{k-1} b b'}$, $Y_{J_{n-1} b f_1}^{I_{k-1} b f_1}$, $Y_{J_{n-1} b' f_1}^{I_{k-1} b' f_1}$. The corresponding equations (from the general formula \eqref{e:W1_formula}) read:
        % \begin{align*}
        %     X_{J_{n-1} b}^{I_{k-1} b} + X_{J_{n-1} b'}^{I_{k-1} b'}=0 \\
        %     X_{J_{n-1} b}^{I_{k-1} b} + X_{J_{n-1} f_1}^{I_{k-1} f_1}=0 \\
        %     X_{J_{n-1} b'}^{I_{k-1} b'} + X_{J_{n-1} f_1}^{I_{k-1} f_1}=0. 
        % \end{align*}
        % In this simple case, it is easy to see that the only possible solution to this system is $X_{J_{n-1} b}^{I_{k-1} b}=X_{J_{n-1} b'}^{I_{k-1} b'}=X_{J_{n-1} f_1}^{I_{k-1} f_1}=0$.

        \item  Suppose now that the component $X_{J_n}^{I_k}$ is generating exactly $\binom{m}{s}$ components of $Y$ one of which is $Y_{J_{n+s}}^{I_{k+s}}$.
        By construction, this means that there are exactly $m$ indexes $b_1 \dots b_m \notin J_{n} \cup I_{k}$. Since $n+k=N-s+p-1$, this implies that $J_{n} \cap I_{k}$ has $q=m-1-s+p$ elements $f_1\dots f_{q}$.
        %$n+k-q+m=N=N-1-s+p-q+m$

        Denoting $J_{n}= J_{n-q} \cup \{f_1 \dots f_q\}$ and $I_{k}= I_{k-q}\cup \{f_1 \dots f_q\}$, we then consider the  $\binom{m+q}{q}$ components $X_{J_{n-q}k_1 \dots k_{q}}^{I_{k-q}k_1 \dots k_{q}}$ one for each possible $q$-element subset $K_{q}$ of $K=\{f_1 \dots f_{q}, b_1 \dots b_m\}$. Each of these components we consider, with a slight abuse of notation the following equations:
        \begin{align*}
            W_{1}^{0 (n,k)}(X_{J_{n-q}k_1 \dots k_{q}}^{I_{k-q}k_1 \dots k_{q}})= \sum_{k_{q+1}\dots k_{q+s} \in K\setminus K_{q}} Y_{J_{n-q}k_1 \dots k_{q}k_{q+1}\dots k_{q+s} }^{I_{k-q}k_1 \dots k_{q}k_{q+1}\dots k_{q+s}}.
        \end{align*}
        The components in the image are $\binom{m+q}{q+s}$. Note that, by Lemma \ref{lem:binomial_dimensions} and applying some simple property of the binomial coefficients we have
        \begin{align*}
            \binom{m+q}{q} \geq \binom{m+q}{q+s}.
        \end{align*}

        By Lemma \ref{lem:solution_system_surj} it is possible to invert this system and get 
            \begin{align*}
            Y_{J_{n-q} k_1 \dots k_{q+s}}^{I_{k-q} k_1 \dots k_{q+s}}= \sum_{i=0}^{q+s} \sum_{\substack{ k_j \dots k_{j+i} \in \{k_1 \dots k_{q+s}\}\\ h_{1} \dots h_{i} \in K \setminus\{k_1 \dots k_{q+s}\} }} C_i X_{J_{n-q}k_1 \dots \widehat{k_j}\dots \widehat k_{j+s+i}\dots k_{q+s} l_1 \dots l_{i}}^{I_{k-q}k_1 \dots \widehat{k_j}\dots \widehat k_{j+s+i}\dots k_{q+s} l_1 \dots l_{i}}.
        \end{align*}
    \end{itemize}
    Hence, for each component $Y_{J_{n-q} k_1 \dots k_{q+s}}^{I_{k-q} k_1 \dots k_{q+s}}$ we can get a linear combination of components of $X$ that generates exactly it, proving that the map is surjective.
\end{proof}

\section{Corollaries}\label{s:corollaries}

Since the results and proofs might seem rather abstract, let us make here some remarks and present some examples.

\begin{remark}
    The results of Theorem \ref{thm:Main_intro} agree with some previous results of which this article is a generalization. In particular for $N=4$ and $s=1$ some results for some particular choices of $(n,k)$ have been proven in \cite{CS2019} and later extend for some combination of $N$ and $s$ in \cite{CCS2020} in the context of Palatini--Cartan gravity (see Section \ref{s:PC_gravity}). We refer to the appendix of this article for an application of the proof presented in this paper in some specific cases. The strategy applied there is similar to the one in this note.
\end{remark}

As an easy corollary of Theorem \ref{thm:Main_intro} we can show that there is a symmetry between the conditions of injectivity and surjectivity.

\begin{corollary}\label{cor:symmetry}
    The map $W_{s}^{l (n,k)}$ defined in \eqref{e:definition_W} is injective if and only if $W_{s}^{l (n',k')}$ is surjective where $n'=N-l-n-s$ and $k'=N-k-s$.
\end{corollary}
\begin{proof}
    Suppose that $W_{s}^{l (n,k)}$ is injective, i.e. $n+k \leq N  -l- s$. Then we have
    \begin{align*}
        n'+k'=2N-l-2s-(n+k)\geq2N-l-2s- (N  -l- s)=N-s
    \end{align*}
    which shows that $W_{s}^{l (n',k')}$ is surjective. Viceversa, let $W_{s}^{l (n',k')}$ be surjective, i.e. $n'+k' \geq N - s$. We get:
    \begin{align*}
        n+k=2N-l-2s-(n'+k') \leq 2N-l-2s- (N-s)= N-l-s
    \end{align*}
    which shows that $W_{s}^{l (n',k')}$ is injective.
\end{proof}

The following corollary shows that for $l=0$ in order to establish the surjectivity or injectivity of a map, it is sufficient to look at the dimensions of domain and codomain.

\begin{corollary}\label{cor:bulk_dimensions}
    Let $l=0$, then $W_{s}^{l (n,k)}$ is injective if and only if 
    \begin{align}
        \dim \Omega_0^{n,k} \leq \dim \Omega_0^{n+s,k+s}
    \end{align}
    and $W_{s}^{l (n,k)}$ is surjective if and only if 
    \begin{align}
        \dim \Omega_0^{n,k} \geq \dim \Omega_0^{n+s,k+s}.
    \end{align}
\end{corollary}
\begin{proof}
    Let us begin from injectivity. The dimension of domain and codomain are 
    \begin{align*}
        \dim \Omega_0^{n,k} = \binom{N}{n}\binom{N}{k} \qquad \dim \Omega_0^{n+s,k+s} = \binom{N}{n+s}\binom{N}{k+s}.
    \end{align*}
    Let $n+k =N-s -r$, $r \in \mathbb{Z}$ and define $m=n+s$. Then we can write $k=N-n-s-r$ and 
    \begin{align*}
        \binom{N}{n}\binom{N}{k}&=\binom{N}{n}\binom{N}{N-n-s-r}=\binom{N}{n}\binom{N}{m+r}\\
        \binom{N}{n+s}\binom{N}{k+s}&=\binom{N}{n+s}\binom{N}{N-r-n}=\binom{N}{m}\binom{N}{n+r}.
    \end{align*}
    Using a well-known formula for the product of the binomial coefficient \cite[p. 171]{GKP89} we get:
    \begin{align*}
        \binom{N}{n}\binom{N}{k}&=\sum_{j=0}^{\mathrm{min}(n,m+r)} \binom{m+n+r-j}{j,m+r-j,n-j}\binom{N}{m+n+r-j}\\
        \binom{N}{n+s}\binom{N}{k+s}&=\sum_{j=0}^{\mathrm{min}(n+r,m)} \binom{m+n+r-j}{j,m-j,n+r-j}\binom{N}{m+n+r-j}
    \end{align*}
    where the first is the multinomial coefficient.
    Suppose now that the function is injective, i.e. $r\geq 0$. Then we have that $n=\mathrm{min}(n,m+r) < \mathrm{min}(n+r,m)$
     Hence we get
    \begin{align*}
        \binom{N}{n+s}\binom{N}{k+s}&-\binom{N}{n}\binom{N}{k} \\
        =& \sum_{j=0}^{n}\left[ \binom{m+n+r-j}{j,m-j,n+r-j}-\binom{m+n+r-j}{j,m+r-j,n-j}\right]\binom{N}{m+n+r-j}\\
        &+ \sum_{j=n}^{\mathrm{min}(n+r,m)} \binom{m+n+r-j}{j,m-j,n+r-j}\binom{N}{m+n+r-j}
    \end{align*}
We now prove that every element in the sum is positive and hence the dimension of the codomain is bigger than the dimension of the domain. By definition of the multinomial coefficient we get
\begin{align*}
    &\binom{m+n+r-j}{j,m-j,n+r-j}-\binom{m+n+r-j}{j,m+r-j,n-j}= \frac{(m+n+r-j)!}{j!(m-j)!(n+r-j)!}-\frac{(m+n+r-j)!}{j!(m+r-j)!(n-j)!}\\
    &= \frac{(m+n+r-j)!}{j!}\frac{(m+r-j)\dots (m+1-j) - (n+r-j)\dots(n+1-j)}{(m+r-j)!(n+r-j)!}\geq 0
\end{align*}
since $m>n$.

In order to prove the converse, let $r<0$. Then $n+r=\mathrm{min}(n+r,m)<\mathrm{min}(n,m+r) $ and we get
\begin{align*}
        \binom{N}{n+s}\binom{N}{k+s}&-\binom{N}{n}\binom{N}{k} \\
        = &\sum_{j=0}^{n+r}\left[ \binom{m+n+r-j}{j,m-j,n+r-j}-\binom{m+n+r-j}{j,m+r-j,n-j}\right]\binom{N}{m+n+r-j}\\
        &- \sum_{j=n+r}^{\mathrm{min}(n,m+r)} \binom{m+n+r-j}{j,m+r-j,n-j}\binom{N}{m+n+r-j}.
    \end{align*}
    Proceeding as before one gets
    \begin{align*}
        \binom{m+n+r-j}{j,m-j,n+r-j}-\binom{m+n+r-j}{j,m+r-j,n-j}\leq 0.
    \end{align*}
    Hence we conclude that the codomain has dimension smaller than the domain.

    For Surjectivity it is sufficient to combine the result of injectivity together with Corollary \ref{cor:symmetry}.
\end{proof}

\begin{remark}
The conclusion of Corollary \ref{cor:bulk_dimensions} are no longer true for higher codimension, $l>1$. As a counterexample, one can consider for $N=4$, $W_2^{1(1,1)}$ which is not surjective despite the dimension of the domain is bigger than the dimension of the codomain. For $l=1$ the proof of Corollary \ref{cor:bulk_dimensions} is no longer true, but no counterexample have been found so far.
\end{remark}

\begin{remark}
    Note that, because of the summand $l$ appearing in the formula in Theorem \ref{thm:Main_intro}, maps that are injective for $l=0$ might not be injective for higher codimension. 
\end{remark}

\section{Coframe of Palatini--Cartan gravity} \label{s:PC_gravity}
In this section we draw the connection between the maps \eqref{e:def_W_intro} and their properties and the coframe formulation of General Relativity, also known as Palatini--Cartan gravity. 

Let $M$ be an $N$-dimensional compact, oriented smooth manifold $M$ and let $\mathcal{V}$ and vector bundle over $M$ such that there is an orientation preserving bundle isomorphism covering the identity from the tangent bundle to $\mathcal{V}$ 
\begin{align*}
    e \colon TM \stackrel{\sim}{\longrightarrow}\mathcal{V}.
\end{align*}

We then consider a $SO(N-1,1)$-principal bundle $P \rightarrow M$ and $\mathcal{V}$ as the associated vector bundle by the standard representation. Then the isomorphism $e$ is one of the fundamental field of the theory, generating the metric, together with a principal connection $\omega$, see \cite{CCS2020} for a more detailed description. 

In field theory it is sometimes useful to consider such theories defined on stratified manifolds, defined as follows.
\begin{definition}\label{def:stratification}
Let $M$ be an $m$-dimensional smooth manifold. An $n$-stratification of $M$ is a filtration of smooth manifolds (possibly with boundary) $\{ M ^{\filt{k}}\}_{k=0\dots n}$ of dimension $\mathrm{dim}(M^{\filt{k}})=m-k$, with $\filtint{0}=M$, such that there exists a smooth embedding $\iota^{\filt{k+1}}\colon M^{\filt{k+1}} \to  M^{\filt{k}}$ for every $0\leq k < n$.
\end{definition}

Let then $M$ be a stratified manifold of dimension $N$ as in Definition \ref{def:stratification}. We will denote by $\mathcal{V}_{l}$  for the restriction $\mathcal{V}|_{M^{(l)}}$ of $\mathcal{V}$ to the $l$-th stratum $M^{(l)}$.
We can then consider 
\begin{align*}
    \Omega_l^{n,k}:= \Omega^n\left(M^{(l)}, \textstyle{\bigwedge^k} \mathcal{V}_l\right) \qquad \text{for } n\leq N-l \text{ and } k \leq N
\end{align*}
and, since $e$ can be naturally be viewed as a one form on $M$ with values on $\mathcal{V}$, i.e. $e \in \Omega^1(M, \mathcal{V})$, we define the maps
\begin{align}\label{e:def_W_intro2}
    W_{s}^{(n,k)}: \Omega^n\left(M, \textstyle{\bigwedge^k} \mathcal{V}\right)  & \longrightarrow \Omega^{n+s}\left(M, \textstyle{\bigwedge^{k+s}} \mathcal{V}\right) \\
    X  & \longmapsto   X \wedge \underbrace{e \wedge \dots \wedge e}_{s-times}. \nonumber
\end{align}

Locally, the coframe these maps are exactly the ones described in the introduction where now $V$ is the typical fiber of $TM$ and $Z$ is the typical fiber of $\mathcal{V}$. Since all the properties, and proofs of the maps \eqref{e:def_W_intro} are local, we conclude that also the maps \eqref{e:def_W_intro2} satisfy the properties described in Theorem \ref{thm:Main_intro} and all its corollaries.

% This isomorphism can be naturally be viewed as a one form on $M$ with values on $\mathcal{V}$, i.e. $e \in \Omega^1(M, \mathcal{V})$. Using this interpretation and the wedge product it is possible to define a class of maps
% \begin{align}\label{e:def_W_intro2}
%     W_{s}^{(n,k)}: \Omega^n\left(M, \textstyle{\bigwedge^k} \mathcal{V}\right)  & \longrightarrow \Omega^{n+s}\left(M, \textstyle{\bigwedge^{k+s}} \mathcal{V}\right) \\
%     X  & \longmapsto   X \wedge \underbrace{e \wedge \dots \wedge e}_{s-times}. \nonumber
% \end{align}
% In the case $M$ is a stratified manifold (see Definition \ref{def:stratification}), it is possible to extend these maps also to the strata of higher codimension. If we denote by $l$ the codimension, the goal of this paper is to give a necessary and sufficient condition for these maps to be injective and/or surjective. Namely, we will prove the following result:
% \begin{theorem}\label{thm:Main_intro2}
%     The map $W_{s}^{l (n,k)}$ defined in \eqref{e:definition_W} is surjective if and only if 
%     \begin{align}
%         n+k \geq N - s
%     \end{align}
%     and it is injective if and only if 
%     \begin{align}
%         n+k \leq N  -l- s . 
%     \end{align}
% \end{theorem}

% \begin{corollary}
%     The map $W_{s}^{l (n,k)}$ is an isomorphism if and only if $l=0$ and $n+k=N-s$.
% \end{corollary}

The importance of these maps is related to the role that they play in the coframe formalism of General Relativity. Indeed, in this case we consider a $SO(N-1,1)$-principal bundle $P \rightarrow M$ and $\mathcal{V}$ is the associated vector bundle by the standard representation. Then the isomorphism $e$ is one of the fundamental field of the theory, generating the metric, together with a principal connection $\omega$, see \cite{CCS2020} for a more detailed description. 

In particular the injectivity or surjectivity of the maps \eqref{e:def_W_intro2} is crucial in order to find the right equation of motion from the action of the theory \cite{AlexandrovSpeziale15, Ashtekar1986, CS2019}. These properties are often used loosely, but adequate care is required when considering the corresponding maps in the case of a manifold with boundary (treated here as a particular case of a stratified manifold. Indeed, as it was pointed out first in \cite{CS2019} for $N=4$ (the physical case), the differences in invertibility of the maps $W_1^{0(2,1)}$ and $W_1^{1(2,1)}$ (respectively acting on the same forms on the bulk and on the boundary) is a crucial aspect of the structure of the reduced phase space of this theory in vacuum \cite{CS2019,CCT2020} or in presence of matter \cite{CCF2022,CCFT2023}. In this latter case, other maps of this type play an equally important role. Furthermore such properties turn out to be important also in the BV(-BFV) formulation of the coframe theory, both on the boundary \cite{CS2017,CCS2020,CCS2020b} and on the corner \cite{CC2023}.

Some particular cases of Theorem \ref{thm:Main_intro} have been proven in \cite{CS2019, CCS2020, C21}, and this note aims to solve once and for all the issues related to these maps.

As a concrete example we can fix $N=4$ (i.e. the physical four-dimensional space-time) and $s=1$.  Then the properties of the maps can be arranged in some tables, one for each codimension.
 We organize the $\Omega_{l}^{n,k}$ spaces in an array and connect them with arrows corresponding to the maps $W_1^{l(n,k)}$: a hooked arrow denotes an injective map while a two headed arrow denotes a surjective map. 

We start with codimension 0, i.e. with the table of the maps in the bulk:
\begin{center}
\begin{equation}
\begin{tikzcd}[ row sep= 2 em, column sep= 3 em]
\Omega^{0,0} \arrow[rd, hook]& \Omega^{0,1} \arrow[rd, hook]& \Omega^{0,2} \arrow[rd, hook]& \Omega^{0,3} \arrow[rd, hook, two heads]& \Omega^{0,4} \\
\Omega^{1,0} \arrow[rd, hook]& \Omega^{1,1} \arrow[rd, hook]& \Omega^{1,2} \arrow[rd, hook, two heads]& \Omega^{1,3} \arrow[rd, two heads]& \Omega^{1,4} \\
\Omega^{2,0} \arrow[rd, hook]& \Omega^{2,1} \arrow[rd, hook, two heads]& \Omega^{2,2} \arrow[rd, two heads]& \Omega^{2,3} \arrow[rd, two heads]& \Omega^{2,4} \\
\Omega^{3,0} \arrow[rd, hook, two heads]& \Omega^{3,1} \arrow[rd, two heads]& \Omega^{3,2} \arrow[rd, two heads]& \Omega^{3,3} \arrow[rd, two heads]& \Omega^{3,4} \\
\Omega^{4,0} & \Omega^{4,1} & \Omega^{4,2} & \Omega^{4,3} & \Omega^{4,4} 
\end{tikzcd}
\end{equation}
\end{center}
On the boundary, i.e. in codimension 1, the index $i$ runs only between 1 and 3. Note that the properties of the maps heavily depend on the codimension in which we are working. We have the following table:
\begin{center}
\begin{equation}
\begin{tikzcd}[ row sep= 2 em, column sep= 3 em]
\Omega_{\partial}^{0,0} \arrow[rd, hook]& \Omega_{\partial}^{0,1} \arrow[rd, hook]& \Omega_{\partial}^{0,2} \arrow[rd, hook]& \Omega_{\partial}^{0,3} \arrow[rd, two heads]& \Omega_{\partial}^{0,4} \\
\Omega_{\partial}^{1,0} \arrow[rd, hook]& \Omega_{\partial}^{1,1} \arrow[rd, hook]& \Omega_{\partial}^{1,2} \arrow[rd, two heads]& \Omega_{\partial}^{1,3} \arrow[rd, two heads]& \Omega_{\partial}^{1,4} \\
\Omega_{\partial}^{2,0} \arrow[rd, hook]& \Omega_{\partial}^{2,1} \arrow[rd, two heads]& \Omega_{\partial}^{2,2} \arrow[rd, two heads]& \Omega_{\partial}^{2,3} \arrow[rd, two heads]& \Omega_{\partial}^{2,4} \\
\Omega_{\partial}^{3,0} & \Omega_{\partial}^{3,1} & \Omega_{\partial}^{3,2} & \Omega_{\partial}^{3,3} & \Omega_{\partial}^{3,4} 
\end{tikzcd}
\end{equation}
\end{center}
In codimension 2 we have the following table:
\begin{center}
\begin{equation}\label{e:cod2We}
\begin{tikzcd}[ row sep= 2 em, column sep= 3 em]
\Omega_{\partial\partial}^{0,0} \arrow[rd, hook]& \Omega_{\partial\partial}^{0,1} \arrow[rd, hook]& \Omega_{\partial\partial}^{0,2} \arrow[rd]& \Omega_{\partial\partial}^{0,3} \arrow[rd, two heads]& \Omega_{\partial\partial}^{0,4} \\
\Omega_{\partial\partial}^{1,0} \arrow[rd, hook]& \Omega_{\partial\partial}^{1,1} \arrow[rd]& \Omega_{\partial\partial}^{1,2} \arrow[rd, two heads]& \Omega_{\partial\partial}^{1,3} \arrow[rd, two heads]& \Omega_{\partial\partial}^{1,4} \\
\Omega_{\partial\partial}^{2,0} & \Omega_{\partial\partial}^{2,1} & \Omega_{\partial\partial}^{2,2} & \Omega_{\partial\partial}^{2,3} & \Omega_{\partial\partial}^{2,4}
\end{tikzcd}
\end{equation}
\end{center}
We conclude with the table for codimension 3, denoting the spaces by $\Omega_{\partial\partial\partial}^{i,j}$ defined in a similar way to the ones in lower codimension. These properties are not used in this thesis but are presented here for future convenience.
\begin{center}
\begin{equation}
\begin{tikzcd}[ row sep= 2 em, column sep= 3 em]
\Omega_{\partial\partial\partial}^{0,0} \arrow[rd, hook]& \Omega_{\partial\partial\partial}^{0,1} \arrow[rd ]& \Omega_{\partial\partial\partial}^{0,2} \arrow[rd]& \Omega_{\partial\partial\partial}^{0,3} \arrow[rd, two heads]& \Omega_{\partial\partial\partial}^{0,4} \\
\Omega_{\partial\partial\partial}^{1,0} & \Omega_{\partial\partial\partial}^{1,1} & \Omega_{\partial\partial\partial}^{1,2} & \Omega_{\partial\partial\partial}^{1,3}& \Omega_{\partial\partial\partial}^{1,4} 
\end{tikzcd}
\end{equation}
\end{center}
\begin{refcontext}[sorting=nyt]
    \printbibliography[] 
\end{refcontext}

\end{document}

%% file: packages.tex
    \usepackage[english]{babel} %% English standard typographical conventions

    \usepackage{graphicx}
    \usepackage[]{color}  %% figures, more colors
    \usepackage[dvipsnames]{xcolor} %% more colors
    \usepackage{csquotes} %% necessary for biber
    \usepackage{appendix} %% for appendices at the end of each chapter
    \usepackage[a4paper]{geometry} %% Standard margins
    \usepackage{cancel} %%strikethrough lines

    \usepackage{amsmath,amssymb,amsfonts,amsthm} %% ams packages
    \usepackage{mathtools,todonotes} %% add-on to amsmath
    \usepackage{tikz-cd}  %%For commutative diagrams with tikz
    \usepackage{tikz}
        \usetikzlibrary{arrows}
    \usepackage[all,cmtip]{xy} %% For diagrams
    \usepackage[bbgreekl]{mathbbol} %%For bold greek letters
    \usepackage{array} %% for arrays

    \usepackage{hyperref} %% hyperlinks
    \usepackage{authblk} %%For affiliation of authors

%% file: technical.tex
\theoremstyle{definition} %%upright text, extra space above and below
    \newtheorem{definition}{Definition}

\theoremstyle{plain} %% italic text, extra space above and below
    \newtheorem{theorem}[definition]{Theorem}
    
    \newtheorem{lemma}[definition]{Lemma}
    \newtheorem{corollary}[definition]{Corollary}

\theoremstyle{remark} %% upright text, no extra space above or below
    \newtheorem{remark}[definition]{Remark}

\usepackage[bibstyle=alphabetic,citestyle=alphabetic,useprefix,giveninits=true, sorting=ynt, sortcites, minbibnames=99,maxbibnames=99,backend=biber]{biblatex}  %%other styles: numeric-comp, authoryear 
\renewbibmacro{in:}{} %% removes in: before journals
\bibliography{bibliography,bibliography2}  %% Name of the file with the bibliography
\DeclareSourcemap{
  \maps[datatype=bibtex]{
    \map[overwrite]{ %%If DOI is present, doesn't print arXiv
      \step[fieldsource=doi, final]
      \step[fieldset=url, null]
      \step[fieldset=eprint, null]
      }
     \map[overwrite]{ %% If only arXiv is present, doens't print pages and eid (eliminates duplicates)
      \step[fieldsource=eprint, final]
      \step[fieldset=pages, null]
      \step[fieldset=eid, null]
      \step[fieldset=journal, null]
    }  
  }
}

     \newcommand{\filt}[1]{(#1)}
     \newcommand{\filtint}[1]{M^{(#1)}}